\documentclass[letterpaper, 10 pt, conference]{ieeeconf}  

\IEEEoverridecommandlockouts                             
\makeatletter
\let\NAT@parse\undefined
\makeatother
\usepackage[square, sort&compress, comma, numbers]{natbib}

\usepackage{graphicx}
\usepackage{subfigure}
\graphicspath{{figures/}{./}} 
\usepackage[latin1]{inputenc}
\usepackage{fancyhdr}    
\usepackage{textcomp}  
\usepackage{amsmath} 
\usepackage{amssymb}
\usepackage{dsfont}
\usepackage{psfrag}
\usepackage{stfloats}
\usepackage{color}
\usepackage{multicol}
\usepackage{algorithmic}
\usepackage{algorithm}
\usepackage{array}
\usepackage{bm}
\usepackage{paralist}{}
\usepackage{wasysym}

\usepackage[bookmarks=true]{hyperref}
\usepackage[hyphenbreaks]{breakurl} 

\usepackage{amsthm}

\newtheorem{proposition}{Proposition}
\newtheorem{remark}{Remark}
\newtheorem{problem}{Problem}
\newtheorem{definition}{Definition}

\newcommand{\mdots}{\ifmmode\mathinner{\ldotp\kern-0.2em\ldotp\kern-0.2em\ldotp}\else.\kern-0.13em.\kern-0.13em.\fi}
\newcommand{\btriangle}{\textrm{\DOWNarrow}}
\newcommand{\wtriangle}{\triangledown}
\newcommand{\bsquare}{\blacksquare}

\newcommand{\bcircle}{\bullet}
\newcommand{\wcircle}{\circ}

\begin{document}

\title{Privacy-Preserving Vehicle Assignment for Mobility-on-Demand Systems}

\author{Amanda Prorok and Vijay Kumar
\thanks{We gratefully acknowledge the support of NSF grant CNS-1521617, and TerraSwarm, one of six centers of STARnet, a Semiconductor Research Corporation program sponsored by MARCO and DARPA. All authors are with the GRASP Laboratory at the University of Pennsylvania, Philadelphia, PA, 19104, USA.
    {\tt\small \{prorok|kumar\}@seas.upenn.edu}}}
\maketitle

\begin{abstract}
Urban transportation is being transformed by mobility-on-demand (MoD) systems. One of the goals of MoD systems is to provide personalized transportation services to passengers. This process is facilitated by a centralized operator that coordinates the assignment of vehicles to individual passengers, based on location data. However, current approaches assume that accurate positioning information for passengers and vehicles is readily available. This assumption raises privacy concerns. In this work, we address this issue by proposing a method that protects passengers' drop-off locations (i.e., their travel destinations). Formally, we solve a batch assignment problem that routes vehicles at obfuscated origin locations to passenger locations (since origin locations correspond to previous drop-off locations), such that the mean waiting time is minimized. Our main contributions are two-fold. First, we formalize the notion of privacy for continuous vehicle-to-passenger assignment in MoD systems, and integrate a privacy mechanism that provides formal guarantees. Second, we present a scalable algorithm that takes advantage of superfluous (idle) vehicles in the system, combining multiple iterations of the Hungarian algorithm to allocate a redundant number of vehicles to a single passenger. As a result, we are able to reduce the performance deterioration induced by the privacy mechanism.
We evaluate our methods on a real, large-scale data set consisting of over 11 million taxi rides (specifying vehicle availability and passenger requests), recorded over a month's duration, in the area of Manhattan, New York. 
Our work demonstrates that privacy can be integrated into MoD systems without incurring a significant loss of performance, and moreover, that this loss can be further minimized at the cost of deploying additional (redundant) vehicles into the fleet.
\end{abstract}



\section{Introduction}
The availability of Location-Based Services (LBS) is transforming a wide variety of applications. This development is being fueled by the increasing use of personal mobile communication devices (smart phones) that are endowed with positioning sensors, such as GPS. Importantly, the availability of precise positioning information in dense urban settings, and the joint decrease in communication costs, has paved the way for mobility-on-demand systems (MoD), such as Lyft~\footnote{http://www.lyft.com/} and Uber~\footnote{http://www.uber.com/}.
The potential of improved urban mobility systems has been largely acknowledged due to the possibility of reducing congestion, vehicle service cost and emissions~\cite{Santi:2014qbv}. Importantly, such services also respond to the needs of individuals, for example by reducing travel cost (through vehicle-sharing) and reducing waiting times (through centralized vehicle coordination)~\cite{Alonso-Mora:2017}.

However, the use of LBS to facilitate MoD services poses a privacy threat to the individual participants. Indeed, vehicles reporting the exact coordinates of a user's drop-off location (travel destination) may reveal sensitive information about the user's habits, and hence, may deter users from using such systems. Consequently, we ask ourselves what were to happen if vehicle locations were not reported \emph{precisely}, but rather \emph{imprecisely}. Indeed, by perturbing the vehicle locations, it is expected that the user will enjoy greater privacy --- at the cost of a loss of service quality. 
Hence, our goal is to propose a solution that protects user travel destinations, thus ensuring privacy, while simultaneously minimizing the loss of MoD service quality.
 
\begin{figure}[t]
\centering
\includegraphics[width=0.8\columnwidth]{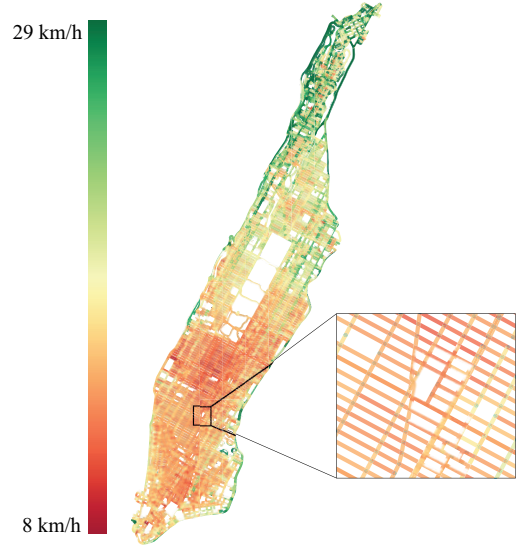}
\caption{Topological representation of Manhattan as computed by our framework described in Sec.~\ref{sec:data}. The edges of the graph are colored to represent the expected traversal speed. We zoom into the area around the Flatiron building, located at 40\textdegree 44' 27.8196'' N 73\textdegree 59' 22.9164'' W.
\label{fig:manhattan}}
\end{figure}

In this work, we consider a fleet of vehicles and passengers demanding to be picked up at specific locations. We pose this problem as a batch assignment of vehicles to passengers, similar to the approach taken in~\cite{Alonso-Mora:2017}. This assignment is facilitated by a centralized operator that collects all customer requests, i.e. the locations at which a vehicle is requested. 
Once a passenger is assigned a vehicle, she communicates her travel destination to her vehicle (by-passing the central operator). Upon completion of the passenger transport, the vehicle immediately communicates its availability to the central operator and specifies its current location. Since doing so would compromise the passenger's travel destination (i.e., the current vehicle position is equal to the dropped-off passenger's destination), we develop an assignment strategy that deals with obfuscated vehicle origin locations.

Although the origin of privacy research stems from the domains of database theory and statistics~\cite{Dwork:2008hs,adam1989security}, it has matured to the point of gaining significant traction in many cross-disciplinary domains, including social networking~\cite{strater2008strategies}, the Internet of Things~\cite{weber2010internet}, and robotics~\cite{Prorok:2016vj}. However, we are unaware of its application to the MoD problem. Indeed, most literature in the domain of MoD systems focuses on the questions of load re-balancing and predictive positioning~\cite{miao2016data, pavone2012robotic,spieser2016shared,miller2016predictive}, and vehicle assignment with passenger pooling~\cite{Santi:2014qbv,Alonso-Mora:2017}. 
 
The contributions of our work are summarized as follows. 
First, we formalize the notion of privacy for continuous vehicle-to-passenger assignment in MoD systems. The key insight is that vehicle origin locations correspond to previous passenger drop-off locations.
Building on prior work in location privacy, we combine the notion of geo-indistinguishability with the batch vehicle assignment problem to ensure private vehicle origin locations. Based on this framework, we quantify the effect of privacy on the performance of the system, as measured by mean passenger waiting times.
Subsequently, we present an algorithm that takes advantage of superfluous vehicles in the system and effectively reduces the performance deterioration induced by the privacy mechanism. 
Our methods are evaluated on a real, large-scale data set consisting of over 11 million taxi rides (specifying vehicle availability and passenger requests), recorded in the area of Manhattan, New York. Our main insight is that the loss of performance induced by privacy is small, and furthermore, this loss can be minimized at the cost of deploying additional vehicles.

\section{Problem Statement}
We consider a batch that consists of $M$ passengers, each requesting one vehicle, and $N$ available vehicles. We model the transport network via a weighted directed graph, $\mathcal{G}= (\mathcal{V},\mathcal{E}, \mathcal{W})$. Vertices in the set $\mathcal{V}$ represent geographic locations, where a node $i$ has a position $\mathbf{x}_i \in \mathbb{R}^2$. 
Nodes $i$ and $j$ are connected by an edge if $(i, j) \in \mathcal{E}$. A weight $w_{ij} \in \mathcal{W}$ quantifies the cost of traversing this edge --- we assume this cost to be equal to the time needed to reach node $j$ from node $i$.
We assume the graph $\mathcal{G}$ is a strongly connected graph, i.e., a path exists between any pair of vertices.
At the beginning of each assignment epoch, vehicles are located at nodes $\mathbf{v} \in \mathcal{V}^{N}$, and passengers are located at nodes $\mathbf{p} \in \mathcal{V}^{M}$. Hence, the positions of a vehicle $i$ and a passenger $j$ are given by $\mathbf{x}_{v_i}$ and $\mathbf{x}_{p_j}$, respectively. 
The vehicle-to-passenger assignment is denoted by a binary matrix $\mathbf{A} \in \{0,1\}^{N \times M}$, which is constrained by $\sum_{i}^N a_{ij} \leq 1$ and $\sum_{j}^{M} a_{ij} \leq 1$ and $\sum_i^{N}\sum_{j}^{M} a_{ij} = \textrm{min}(N,M)$. In other words, we assign $D=1$ vehicles to each passenger, and call this our \emph{non-redundant} scheme~\footnote{For a \emph{redundant} assignment with $D > 1$ (and where $N > M$), we have $\sum_i^N a_{ij} \leq 1$ and $\sum_j^M a_{ij} \geq 1$ and $M \leq \sum_i^N\sum_j^M a_{ij} \leq N$. The performance is then measured as the mean waiting time until pick-up by the fastest vehicle.}.
We capture the \emph{cost} for vehicle $i$ to travel to a passenger $j$ by a matrix $\mathbf{C} \in \mathbb{R}^{N \times M}$ with elements $c_{ij}$. Finally, we measure the performance of our assignment strategy by considering the \emph{waiting times} until pick-up, given by 
$c_{ij}$ where $a_{ij}=1$ for all passengers $j \in 1,\ldots,M$.

Once a vehicle has picked-up and transported a passenger to her desired location, the vehicle notifies the central operator that it is available by communicating its obfuscated position. This position corresponds to the vehicle's origin for the subsequent assignment epoch, with the true (non-obfuscated) value equal to the previous passenger's travel destination.
Our problem can now be stated as follows.
\begin{problem}\label{prob:private_origin} Design a method that routes vehicles to passengers, while minimizing average passenger waiting times, and while guaranteeing a desired level of privacy for passenger travel destinations. 
\end{problem}
%

\section{Manhattan Taxicab Dataset} 
\label{sec:data}
We focus the evaluations of our work on the geographical area of Manhattan, and rely on a public dataset of New York City yellow taxicab operation to provide us with real passenger demand and vehicle availability information~\footnote{NYC Taxi \& Limousine Commission, Trip Record Data, \url{http://www.nyc.gov/html/tlc/html/about/trip_record_data.shtml}}. 
The dataset was collected during the month of June, 2016, and consists of 11 million taxi rides. The data specifies the time and location of pick-up and drop-off, as well as trip distance and fare. 
In order to facilitate the evaluation of our methods, we create a graph of Manhattan by accessing actual street networks from OpenStreetMap~\footnote{\url{http://www.openstreetmap.org}}~\cite{Boeing:2016}. 
Our topological representation of Manhattan consists of 4302 nodes and 9414 edges. In order to deploy algorithms based on this representation, we first define the \emph{cost} of traversing any edge in this graph. In the context of transportation, an intuitive cost function is given by the expected travel time. Hence, we use the pick-up and drop-off locations listed in the June, 2016, dataset to compute all trajectories taken, assuming that the shortest path (length-wise) was chosen. We associate each trajectory with the listed travel time. Each edge of the trajectory is assigned a travel time proportional to its length. After processing all trajectories, we can compute the mean travel time $w_{ij}$ of each edge $(i,j)$ in the graph. 
Figure~\ref{fig:manhattan} shows the resulting expected travel times for all edges in the graph of Manhattan.

In order to solve the assignment problem, in the remainder of this work, we assemble passenger requests and vehicle availabilities into \emph{batches} that consider 20\,second time-windows. Figure~\ref{fig:manhattan_data} shows data collected on Friday June 1st, 2016. We process the ride data to show the number of new passenger requests per batch. The data shows how demand peaks during the morning and late afternoon rush hours, with fluctuations at lunch time.

\begin{figure}[tb]
\centering
\psfrag{t}[cc][][0.8]{Time}
\psfrag{o}[cc][][0.8][90]{Occupied Taxis}
\psfrag{r}[cc][][0.8][-90]{Number of Requests per Batch}
\includegraphics[width=0.8\columnwidth]{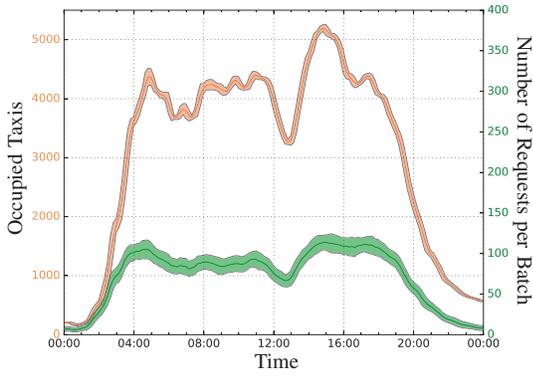}
\caption{Analysis of Manhattan taxicab dataset for data collected on Friday June 1st, 2016. From all rides recorded in that 24\,h time-interval, we only select rides that start and end on the island of Manhattan. 
We show the number of passenger pick-ups made per 20\,s intervals (green curve). We also show the total number of occupied taxis at any given moment (blue curve). The data is smoothed over 30\,min rolling windows, and the shaded areas show the corresponding standard deviations.
\label{fig:manhattan_data}}
\end{figure}

\section{Background} 
\label{sec:preliminaries}

In the following, we review the notions of location privacy upon which we build our methodology. 
Several approaches to location privacy have been proposed thus far --- a comprehensive review is offered in~\cite{Andres:2013geo,koufogiannis2016location}. Most of these methods, however, assume that the adversary's prior belief (side information) is known, and are explicitly modeled on this assumption~\cite{shokri2012protecting}. Such approaches have the downside that any inconsistency or change in the attacker's side information leads to an immediate threat (and privacy is no longer guaranteed).
Indeed, a much stronger definition of privacy is one that is independent of any current or future attacker model.
Consequently, there has been much interest in \emph{differentially private} formalisms that abstract from adversary's side information~\cite{Dwork:2008hs}.

\subsection{Differential Privacy}
\label{sec:diff_privacy}
Stemming from the domain of statistical databases, the goal of differential privacy is to protect individual entries in a given database (in our case, passenger drop-off locations), while simultaneously allowing aggregate information about the database to be released through a query (in our case, a query that outputs the vehicles' origin locations). The key requirement is that changing an individual's entry in the database (i.e., a vehicle origin location that corresponds to a specific passenger drop-off location) should not have a significant affect on the outcome of the query. More formally, if the probability that a query returns a value from a database lies within an $e^\epsilon$ multiplicative bound of the probability that the same query returns the same value from an adjacent database~\footnote{Two databases are \emph{adjacent} if they differ by one entry.}, then the query is said to produce $\epsilon$-indistinguishable outcomes~\cite{Dwork:2006wc}. Notably, this definition is void of any threat model, and hence, is independent of any side information that the attacker might own. In order to preserve $\epsilon$-indistinguishability, privacy mechanisms consist of adding random noise (commonly drawn from a Laplace distribution) to the query output. 

\subsection{Geo-Indistinguishability}
\label{sec:geo_privacy}
The location privacy formalism put forward by Andres et al.~\cite{Andres:2013geo}, termed geo-indistinguishability, is a generalization of differential privacy to the metric domain. In the following, we introduce the main concepts with an adapted notation.
Geo-indistinguishability considers a query that exposes a position $\mathbf{x}$ from a database. The privacy leakage can be formulated as 
\begin{equation}
\mathcal{L} = \sup_{\mathbf{x},\mathbf{x}'} ~  \left | \mathrm{ln} \frac{\mathbb{P}(\tilde{\mathbf{x}}|\mathbf{x})}{\mathbb{P}(\tilde{\mathbf{x}}|\mathbf{x}')} \right |
\end{equation}
where $\mathbf{x}$ is a true position stored in the original database, $\mathbf{x}'$ is the corresponding altered position stored in an adjacent database, and $\tilde{\mathbf{x}}$ is an obfuscated position.
The idea of geo-indistinguishability is to ensure that two positions $\mathbf{x}$ and $\mathbf{x}'$ are indistinguishable when they are close to each other. In other words, a user enjoys $\epsilon r$-privacy within a radius $r$, if any two locations that are at most $r$ apart produce query results with similar distributions.
\begin{definition}[Adapted from Def. 3.1~\cite{Andres:2013geo}: Geo-indistinguishability]
A mechanism that returns $\tilde{\mathbf{x}}$, for a given $\mathbf{x}$ or a given $\mathbf{x}'$, satisfies $\epsilon$-geo-indistinguishability iff for all $\mathbf{x}$ and $\mathbf{x}'$:
\begin{equation}
\mathcal{L} = \sup_{\mathbf{x},\mathbf{x}'} ~  \left | \mathrm{ln} \frac{\mathbb{P}(\tilde{\mathbf{x}}|\mathbf{x})}{\mathbb{P}(\tilde{\mathbf{x}}|\mathbf{x}')} \right | \leq \epsilon ||\mathbf{x}-\mathbf{x}'||_2  
\end{equation}
\end{definition}
Building on prior results~\cite{chatzikokolakis2013broadening}, the authors argue that
the obfuscated position $\tilde{\mathbf{x}}$ is to be drawn from a two-dimensional Laplace distribution inversely scaled by $\epsilon$, and centered at $\mathbf{x}$. Formally, we have that $\tilde{\mathbf{x}} \sim L(\mathbf{x}, \epsilon)$, and we define the corresponding probability density function as $\mathbb{P}_L(\tilde{\mathbf{x}} | \mathbf{x}, \epsilon)$.
In order to satisfy $\epsilon$-geo-indistinguishability, we implement this proposed privacy mechanism~\footnote{We note that Th. 4.1 of~\cite{Andres:2013geo} proves that under double precision with 16 significant digits, the discretization of noisy data points onto a grid does not incur a loss of privacy.}.

\newcommand{\fc}{0.33} 
\begin{figure*}[tb]
\centering
\subfigure[$\epsilon = 0.005$]{\includegraphics[width=\fc\columnwidth]{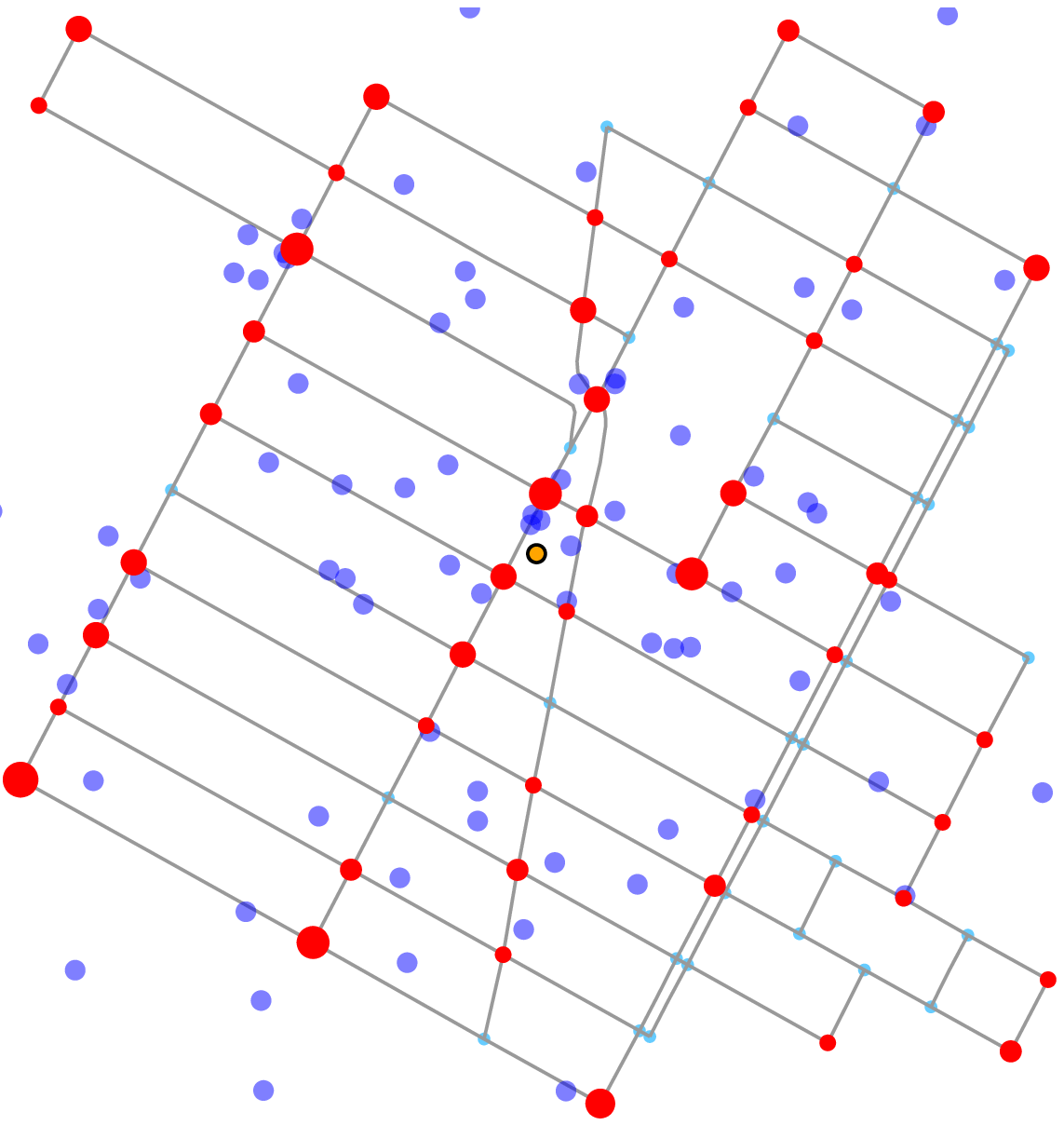}\label{fig:manhattan_noisy_a}}\hspace{0.2cm}
\subfigure[$\epsilon = 0.01$]{\includegraphics[width=\fc\columnwidth]{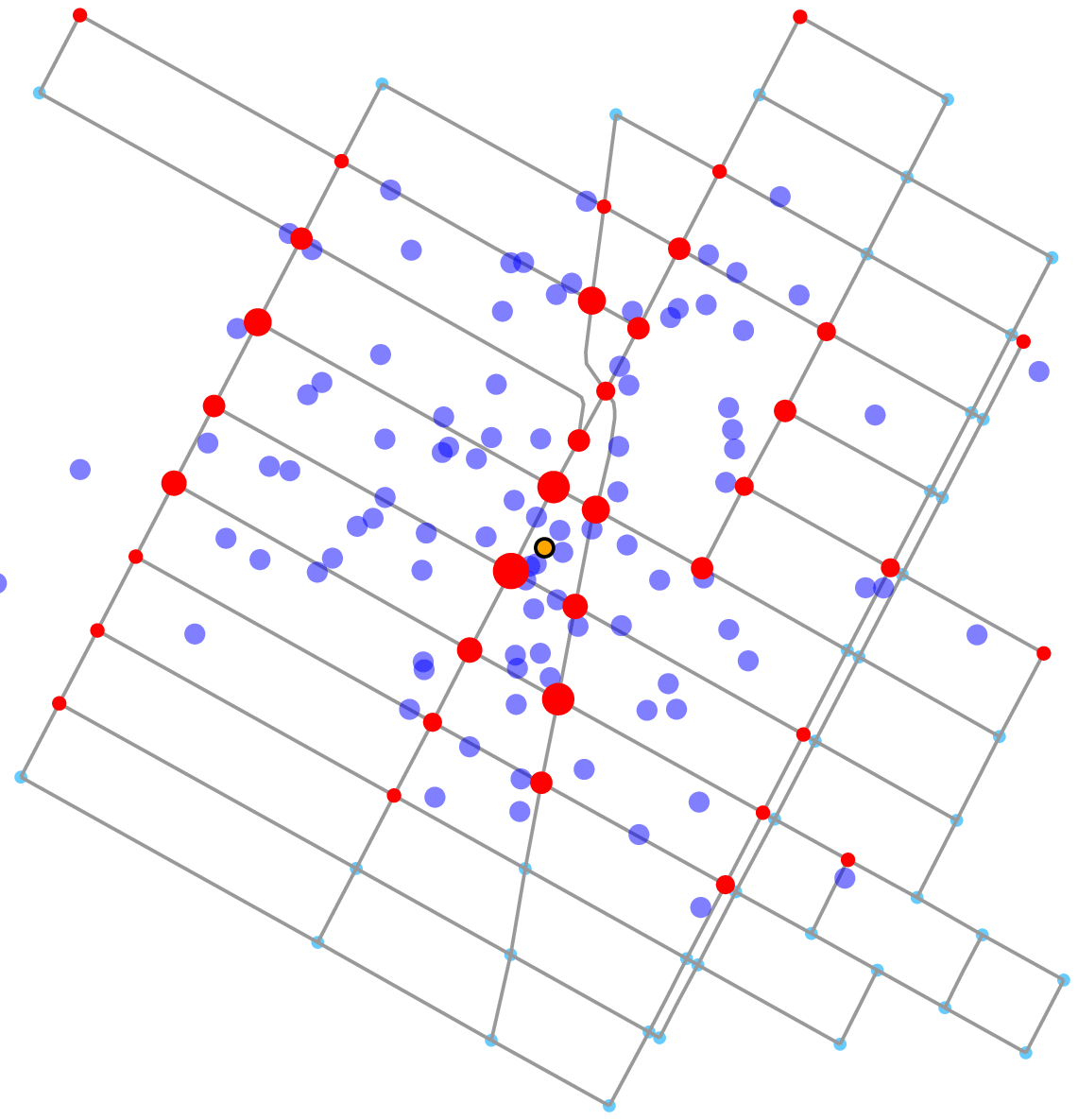}\label{fig:manhattan_noisy_b}}\hspace{0.2cm}
\subfigure[$\epsilon = 0.02$]{\includegraphics[width=\fc\columnwidth]{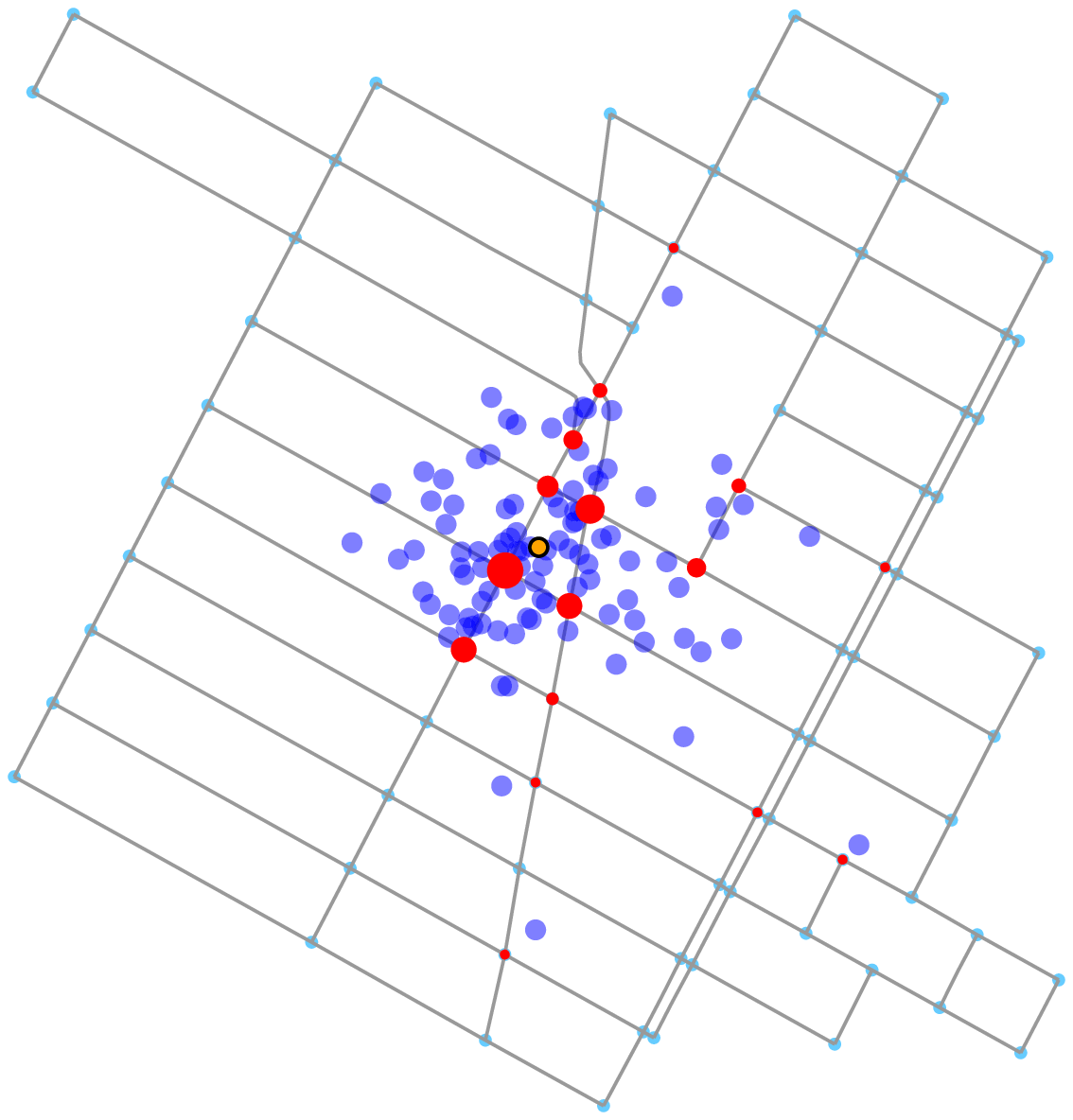}\label{fig:manhattan_noisy_c}}\hspace{0.2cm}
\subfigure[$\epsilon = 0.05$]{\includegraphics[width=\fc\columnwidth]{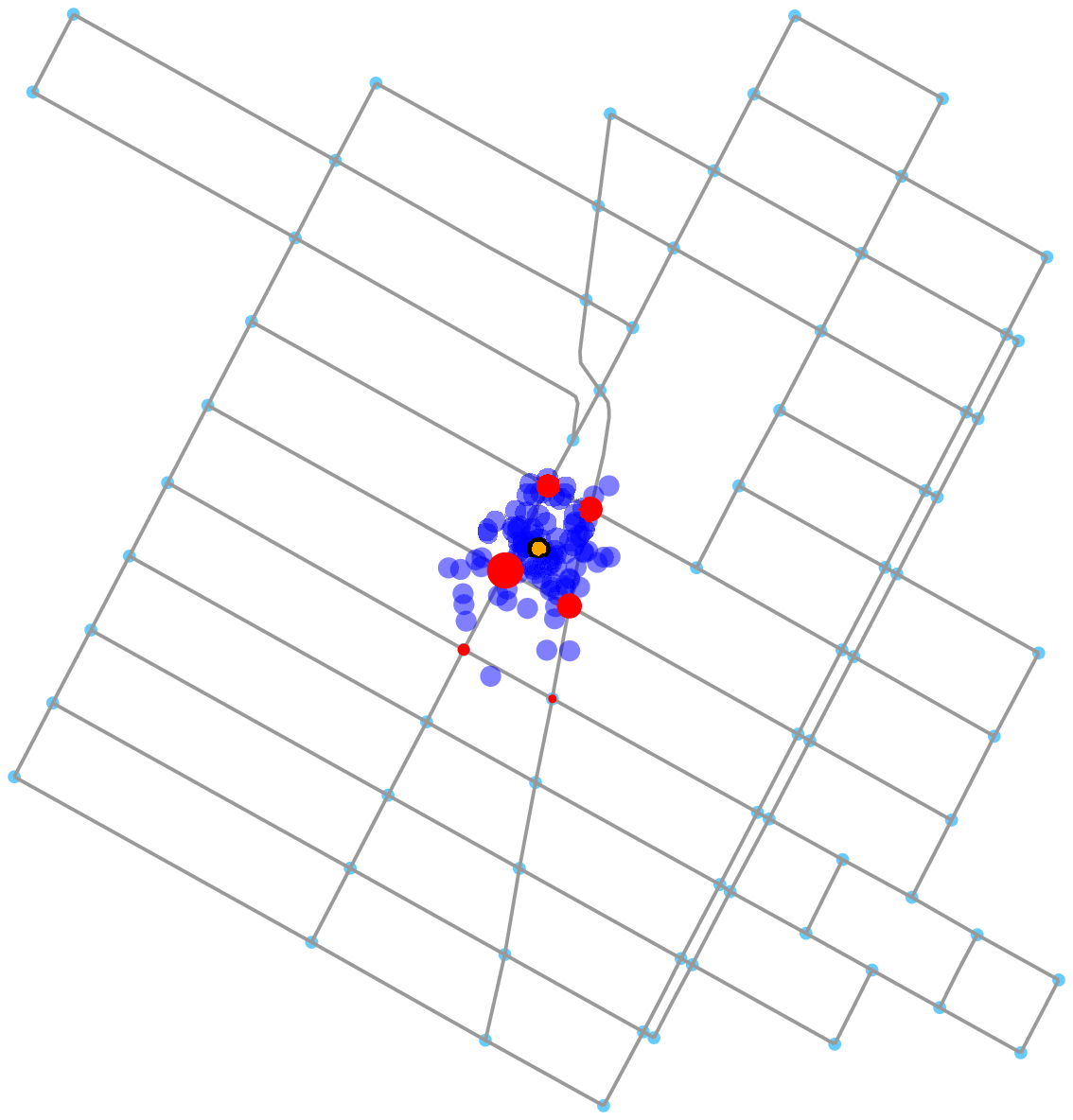}\label{fig:manhattan_noisy_d}}\hspace{0.2cm}
\subfigure[$\epsilon = 0.1$]{\includegraphics[width=\fc\columnwidth]{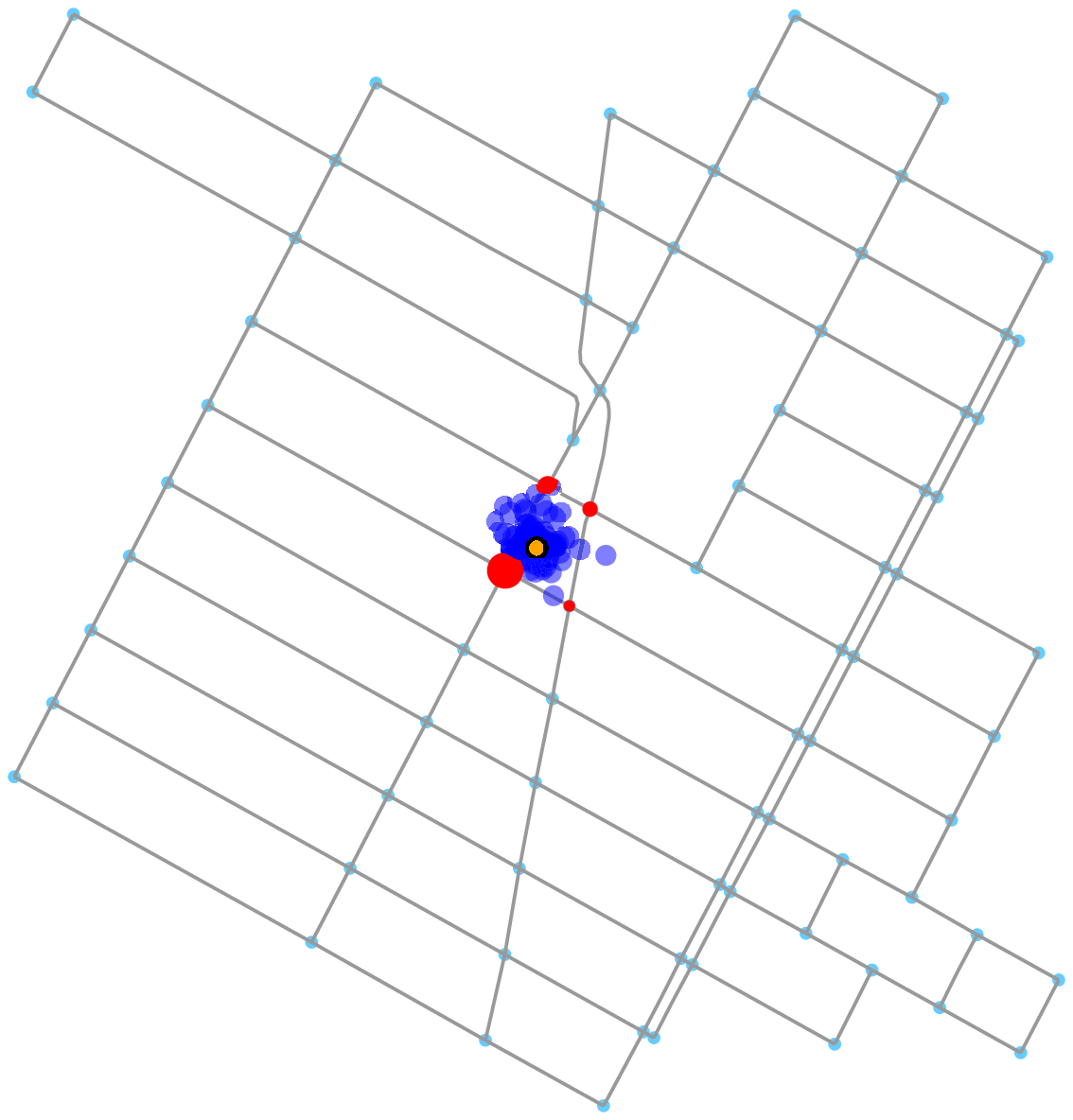}\label{fig:manhattan_noisy_e}}
\caption{A sub-area of Manhattan, centered around the Flatiron building located at 40\textdegree 44' 27.8196'' N 73\textdegree 59' 22.9164'' W. We draw 100 position samples from a two-dimensional Laplacian with inverse scale $\epsilon$, centered around the Flatiron building. The random samples are subsequently projected to the nearest vertices on the associated graph (shown in red), where the size of the node corresponds to the multiplicity of projected samples at that node.
\label{fig:manhattan_noisy}}
\end{figure*}

Fig.~\ref{fig:manhattan_noisy} demonstrates the effect of this mechanism, applied to the coordinates of the Flatiron building in Manhattan. 
We observe how, as the scale of the Laplacian increases (i.e., $\epsilon$ decreases), the noise (and hence privacy) increases. 
In the context of vehicle routing, it becomes clear that increased privacy comes at the cost of performance deterioration due to an obfuscation of vehicle positions that leads to suboptimal vehicle routing.
In the following sections, we discuss this effect and propose a method that enables a minimization of this loss of performance.


\section{Batch Vehicle Routing under Privacy} 
\label{sec:det_routing}

The goal is to assign and route vehicles to passengers such that each passenger is picked up, while minimizing the total assignment cost. We formalize this vehicle routing problem as finding the optimal assignment solution $\mathbf{A}^{\star}$:
\begin{eqnarray}
  \mathbf{A}^{\star} = {\mathop {\mathrm{argmin}} \limits_{\mathbf{A}} } && \sum_{i=1}^{N}\sum_{j=1}^{M} c_{ij} a_{ij} \label{eq:vrp} 
\end{eqnarray}
with constraints $\sum_{i}^N a_{ij} \leq 1$ and $\sum_{j}^{M} a_{ij} \leq 1$ and $\sum_i^{N}\sum_{j}^{M} a_{ij} = \textrm{min}(N,\allowbreak M)$. The element $a_{ij}^{\star}$ of matrix $\mathbf{A}^{\star}$ specifies whether the final solution routes vehicle $i$ to passenger $j$.

The system above is a linear sum assignment problem, also known as the problem of minimum weight matching in bipartite graphs. We use the Hungarian algorithm (or Kuhn-Munkres algorithm), to solve the system and find an optimal assignment $\mathbf{A}^{\star}$.
This assignment is deterministic, and vehicles follow the shortest path (or one of the shortest paths, if several exist) to reach their assigned passenger.

To compute the elements $c_{ij}$ of the cost matrix $\mathbf{C}$, we consider the cost incurred when routing a vehicle located at a node $i$ to a passenger located at a node $j$. The cost of this path is given by the sum of the weights of edges that lie on it
\begin{equation}
\label{eq:cost}
f(i,j) = \sum_{(k,l) \in \mathcal{S}_{ij}} w_{kl},
\end{equation}
where $\mathcal{S}_{ij}$ is the set of edges in the shortest path between node $i$ and node $j$, and $w_{kl} \in \mathcal{W}$ is the weight of an edge $(k,l)$.
We can now compute the cost for all possible vehicle-to-passenger assignments
\begin{equation}
c_{ij} = f(v_i, p_j),~ \forall i, j
\end{equation}
and subsequently solve system~\eqref{eq:vrp}.

\subsection{Solving the Assignment Problem under Obfuscation}

Our goal is to increase the privacy of vehicle origin locations (we remind the reader that the vehicle origin and the previous passenger drop-off locations are the same). We do this by implementing the privacy mechanism described in Sec.~\ref{sec:geo_privacy} to produce obfuscated (noisy) vehicle origin locations, denoted by $\mathbf{x}_{\tilde{v}_i}$ for all vehicles $i=1,\ldots,N$ --- i.e., $\mathbf{x}_{\tilde{v}_i} \sim L(\mathbf{x}_{v_i}, \epsilon)$. We compute the expected cost $\tilde{c}_{ij}$ of routing a vehicle from a probable node $v_i$ to a true passenger location ${p}_j$, given that the vehicle is located around a noisy position ${\mathbf{x}_{\tilde{v}_i}}$ generated by a planar Laplace distribution with inverse scale parameter $\epsilon$:
\begin{equation}
\label{eq:exp_cost}
\tilde{c}_{ij} = \mathbb{E}[c_{ij}] = \eta \sum_{k \in \mathcal{V}} \mathbb{P}_L(\mathbf{x}_{\tilde{v}_i} | k, \epsilon) f(k, p_j).
\end{equation}
where $\eta$ is a normalization constant. 

We adapt the original objective in~\eqref{eq:vrp} to account for the expected cost:

{\small
\begin{eqnarray} \label{eq:vrp_exp} 
  \mathbf{\tilde{A}}^{\star} = {\mathop {\mathrm{argmin}} \limits_{\mathbf{A}} } ~\mathbb{E} \left[ \sum_{i=1}^{M}\sum_{j=1}^{N} c_{ij} a_{ij} \right] = {\mathop {\mathrm{argmin}} \limits_{\mathbf{A}} } \sum_{i=1}^{M}\sum_{j=1}^{N} ~\mathbb{E}[ c_{ij}] a_{ij} .
\end{eqnarray}
}

We note that, since the cost values $\tilde{c}_{ij}$ are noisy, this assignment produces a \emph{suboptimal} assignment $\tilde{\mathbf{A}}^{\star}$ with respect to the true vehicle origin locations. We measure the performance of this assignment by considering the passenger waiting times $c_{ij}$ with $\tilde{a}^{\star}_{ij}=1$, where $c_{ij}$ corresponds to the effective waiting time (based on the true, non-obfuscated vehicle origins).

\begin{proposition}[$\epsilon$-geo-indistinguishable batch assignment]
\label{prop:1}
The batch assignment where each vehicle $i$ reports an obfuscated position $\mathbf{x}_{\tilde{v}_i}$ drawn from a planar Laplace distribution is $\epsilon$-geo-indistinguishable with respect to the true positions ${\mathbf{x}}_{v_i}$.
\end{proposition}
\begin{proof}
Given a query that reports the current set of vehicle positions, the leakage formula can be written as:
\begin{equation}
\mathcal{L} = \mathop{\textrm{sup}}_{i, \mathbf{x}_{v_i}, \mathbf{x}_{v'_i}} \left| \ln \frac{\mathbb{P}(\mathbf{x}_{\tilde{v}_1}, \mdots, \mathbf{x}_{\tilde{v}_N} | \mathbf{x}_{v_1}, \mdots, \mathbf{x}_{v_N})}{\mathbb{P}(\mathbf{x}_{\tilde{v}_1}, \mdots, \mathbf{x}_{\tilde{v}_N} | \mathbf{x}_{v_1}, \mdots, \mathbf{x}_{v'_i}, \mdots, \mathbf{x}_{v_N})} \right|
\end{equation}
where $\mathbf{x}_{v'_i}$ represents an alternative position for vehicle $i$. The numerator refers to the database containing all true positions, while the denominator refers to an adjacent database where the position of a single vehicle has been changed. By the definition of $\epsilon$-geo-indistinguishability (cf. Section~\ref{sec:geo_privacy}) and since all obfuscated positions are independent, we obtain:
\begin{eqnarray}
\mathcal{L} &=& \mathop{\textrm{sup}}_{i, \mathbf{x}_{v_i}, \mathbf{x}_{v'_i}} \left| \ln \frac{\mathbb{P}(\mathbf{x}_{\tilde{v}_i} | {\mathbf{x}}_{v_i})}{\mathbb{P}(\mathbf{x}_{\tilde{v}_i} | \mathbf{x}_{v'_i})} \right| \\
&\leq& \epsilon \mathop{\textrm{sup}}_{i, \mathbf{x}_{v_i}, \mathbf{x}_{v'_i}} \| \mathbf{x}_{v_i} - \mathbf{x}_{v'_i} \|_2 \nonumber 
\end{eqnarray}
\end{proof}

\subsection{Redundant Vehicle Assignment}
\label{sec:multi_vehicle_ass}

\newcommand{\nneg}{\!\!\!\!\!\!\!\!\!\!\!\!\!\!\!\!}
\begin{algorithm}[tb]
\caption{{\small Iterative Hungarian Assignment under Obfuscation}}
\label{alg1}
\begin{algorithmic}[1]
    \STATE $\tilde{c}_{ij} \Leftarrow \eta \displaystyle\sum_{k \in \mathcal{V}} \mathbb{P}_{L}(\mathbf{x}_{\tilde{v}_i} | k, \epsilon) f(k, p_j) ~ \forall i,j$ \label{alg1:line1}
    \STATE $\tilde{\mathbf{A}}^{\star} \Leftarrow \displaystyle\textrm{argmin}_{\mathbf{A}} \sum_{i = 1}^N \sum_{j = 1}^M \tilde{c}_{ij} a_{ij}$ \label{alg1:line2}
    \FOR{$\{2, \ldots, D\}$} \label{alg1:line3}
        \IF{$N < MD$}
            \STATE {\bf break}
        \ENDIF
        \STATE $\mathcal{Z}_j \Leftarrow \{ i | a^{\star}_{ij} = 1 \}$ \label{alg1:line7}
        \STATE %
            \sbox0{$\vcenter{\hbox{$\begin{array}{l}
                +\infty \quad \textrm{if~} \exists j' \textrm{~s.t.~} i \in \mathcal{Z}_{j'} \\
                \eta \!\!\!\! \displaystyle\sum_{\substack{~\\ \mathbf{k} \in \mathcal{V}^{|\mathcal{Z}_j| + 1}}} \!\!\!\! \min_{k \in \mathbf{k}} f(k, p_j) \nneg \prod_{\{k, i'\} \in \mathrm{zip}(\mathbf{k}, \mathcal{Z}_j \cup \{i\})
}\nneg \mathbb{P}_L(\mathbf{x}_{\tilde{v}_{i'}} | k, \epsilon)
            \end{array}$}}$}%
    ${\tilde{c}}_{ij} \Leftarrow \left\{ \vrule width0pt depth \dimexpr\dp0 + .3ex\relax\copy0  
            \kern-\nulldelimiterspace
    \vphantom{\copy0}\right. 
    $  \label{alg1:line8}
        \STATE $\tilde{\mathbf{A}}^{\star} \Leftarrow \displaystyle \tilde{\mathbf{A}}^{\star} + \textrm{argmin}_{\tilde{\mathbf{A}}} \sum_{i = 1}^N \sum_{j = 1}^M \tilde{c}_{ij} a_{ij}$ \label{alg1:line9}
    \ENDFOR
\end{algorithmic}
\end{algorithm}


Vehicle-to-passenger assignments that are based on obfuscated positions will result in degraded performance. Much of this performance loss can be recovered by realizing that, in practice, a large proportion of the vehicle fleet is idle~\footnote{\url{http://www.nyc.gov/html/tlc/downloads/pdf/2014_taxicab_fact_book.pdf}}. Our idea is to assign \emph{redundant} vehicles to each passenger: of all assigned vehicles, only the fastest vehicle will actually pick up the passenger. Consequently, this strategy reduces the expected passenger waiting time (with respect to the non-redundant assignment strategy)~\footnote{The underlying reasoning is that for two random variables $X$ and $Y$ representing passenger pick-up times, we have that $\mathbb{E}[\mathrm{min}(X,Y)] \leq \mathrm{min}(\mathbb{E}[X],\mathbb{E}[Y])$.}.

Algorithm~\ref{alg1} proposes a polynomial-time procedure that assigns $D$ vehicles to each passenger. When $D > 1$, we refer to the assignment as \emph{redundant}. The key component of this algorithm is that it computes the optimal assignment of (several) idle vehicles to each passenger based on a cost matrix that is built incrementally (with each additionally assigned vehicle).  
Lines~\ref{alg1:line1}~and~\ref{alg1:line2} compute the solution to the basic non-redundant assignment, as seen in the previous section. 
At each iteration (starting on line~\ref{alg1:line3}), the procedure adds an additional vehicle to each passenger such that the expected sum of waiting times is minimized. 
On line~\ref{alg1:line7}, the set $\mathcal{Z}_j$ contains the indeces of the currently assigned vehicles for passenger $j$.
Line~\ref{alg1:line8} computes the expected waiting time resulting from assigning an additional vehicle $i$ to each passenger $j$~\footnote{$\mathrm{zip}(\mathcal{A}, \mathcal{B})$ corresponds to the list of pairs obtained by combining elements of $\mathcal{A}$ and $\mathcal{B}$ in the same order (with $|A| = |B|$). E.g., $\mathrm{zip}(\{\{1, 2\}, \{3, 4\}\}) = \{\{1, 3\}, \{2, 4\}\})$.}. 
%
It does so by evaluating the joint probability that all already assigned vehicles and the additional vehicle are located at a given set of nodes $\mathbf{k}$, and multiplying this probability by the minimum waiting time (given by the node in $\mathbf{k}$ closest to the passenger).
Line~\ref{alg1:line9} combines the previous assignment with the newly optimized one. It is worth noting that line~\ref{alg1:line8} can be computed quickly
\emph{(i)} by memorizing the results of the previous iteration for the next,
\emph{(ii)} by ignoring nodes that have a minor impact on the computation of $\tilde{c}_{ij}$ (e.g., nodes $k$ such that $\mathbb{P}_L(\mathbf{x}_{\tilde{v}_i} | k, \epsilon) \leq p_{\min}$ for some arbitrary threshold $p_{\min}$, and
\emph{(iii)} by pre-computing for each node in the graph, this list of relevant nodes (nearest nodes given the latter threshold), and their shortest route lengths to every other node in the graph.
Hence, the overall complexity is bounded by the Hungarian algorithm, and is in the order of $\mathcal{O}((\min(M^2N, N^2M) + MNDs(\mathcal{V}))D)$ where $s(\mathcal{V}) = \max_{k\in\mathcal{V}} |\{l | l \in \mathcal{V} \vee \mathbb{P}_L(l | k, \epsilon) > p_{\min}\}|$ represents the size of largest set of vertices that contribute minimally (as determined by $p_{\min}$) to the computation of the expected waiting times. For example, setting $p_{\min} = 10^{-6}$ with $\epsilon = 0.02$ results in $s(\mathcal{V}) = 30$ on the Manhattan graph.

\begin{figure}[tb]
\centering
\psfrag{t}[cc][][0.8]{Time [s]}
\psfrag{f}[cc][][0.8][90]{Frequency}
\subfigure[{\footnotesize Optimal\,($\epsilon \!\!\!\rightarrow \!\!\infty$)}]{\includegraphics[width=0.26\columnwidth]{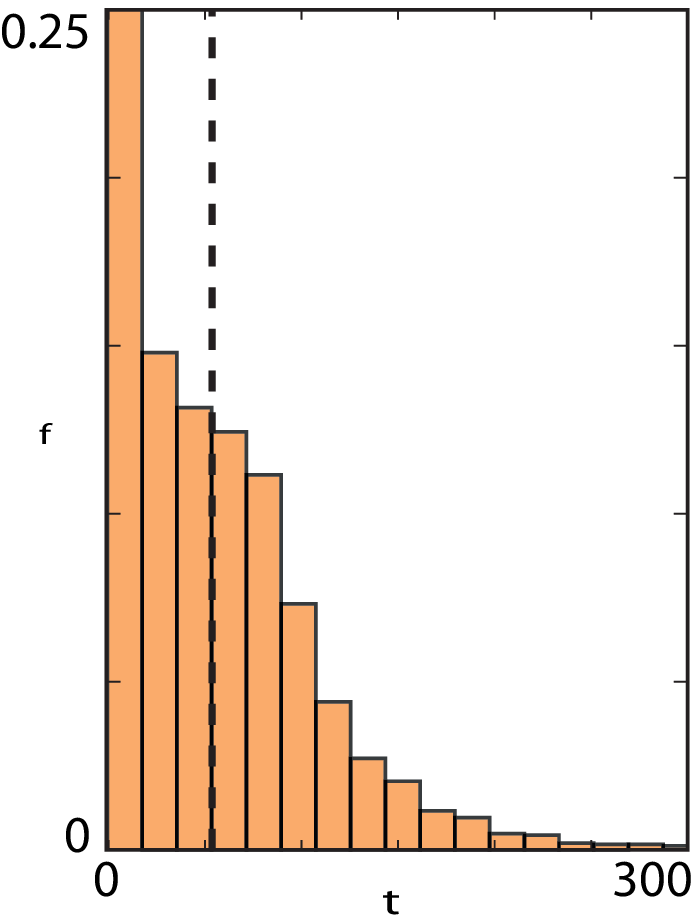}}\hspace{0.1cm}
\subfigure[$\epsilon=0.02$]{\includegraphics[width=0.26\columnwidth]{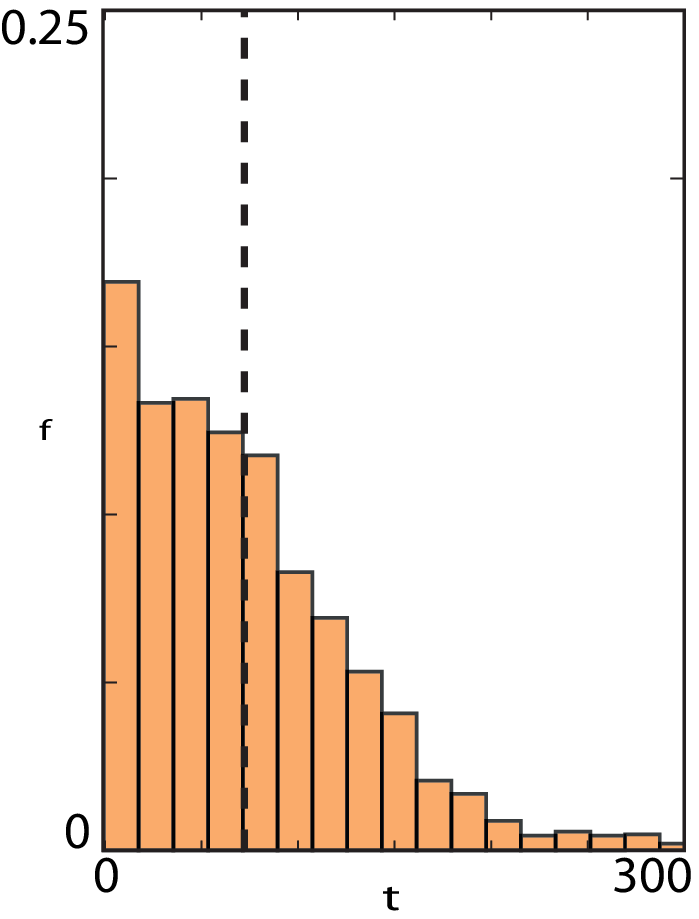}}\hspace{0.1cm}
\subfigure[$\epsilon=0.01$]{\includegraphics[width=0.26\columnwidth]{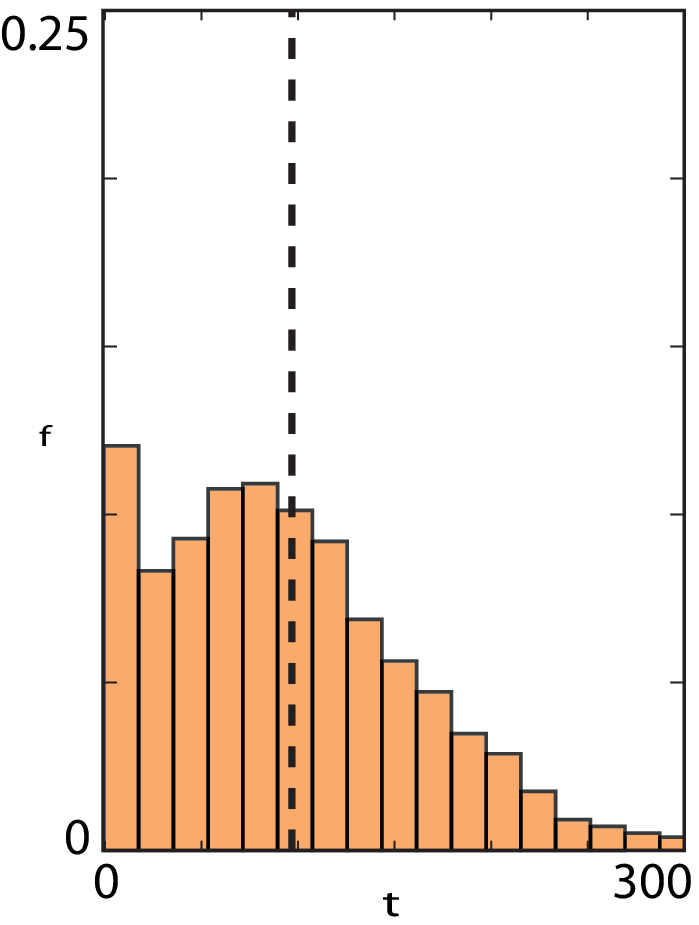}}
\caption{Passenger waiting times for batch vehicle-to-passenger assignment, for 500 vehicles and 250 passengers. Passenger and vehicle locations are sampled according to the actual distribution of pick-up and drop-off locations, respectively, over the month of June 2016. The dashed line shows the mean. (a) Optimal assignment strategy. (b) Private assignment using geo-indistinguishable vehicle origins, with $\epsilon=0.02$. (c) Private assignment using geo-indistinguishable vehicle origins, with $\epsilon = 0.01$.
\label{fig:batch_histogram}}
\end{figure}

\subsection{Performance}
The following results are based on the dataset and graph described in Section~\ref{sec:data}, and show the performance of the batch assignment strategy for varying levels of noise, and a varying number of available vehicles.
Fig.~\ref{fig:batch_histogram} shows passenger waiting times for 500 vehicles and 250 passengers, obtained after non-redundant single-vehicle assignments. Passenger and vehicle locations are sampled according to the actual distribution of pick-up and drop-off locations, respectively, as recorded over the month of June 2016. Using an optimal (noise-free) assignment algorithm, the mean waiting time is just under 1 minute. We observe that as the noise level increases, the distribution shifts, resulting in higher mean waiting times.

Fig.~\ref{fig:batch_waiting_time} shows the performance of the batch assignment algorithm, as a function of the Laplace inverse scale parameter $\epsilon$, for a fixed number of 250 passengers. We consider non-redundant as well as redundant assignments (the number of vehicles assigned per passenger is denoted by $D$).
For all panels, the left side shows the average waiting time, and the right side shows the degradation in performance between the  private (suboptimal) assignment and the non-private (optimal) assignment (also shown by a dashed line on the left panel).
Figures~\subref{fig:batch_waiting_time_a}, and \subref{fig:batch_waiting_time_c} use 250, and 1000 vehicles respectively. 
As expected, the mean waiting time decreases as the number of available vehicles increases. Consequently, as the proportion of available vehicles to passenger increases, the performance of the private assignment strategy deviates more strongly from the optimal performance (since the noise is constant, and the vehicle density increases).

\newcommand{\fccc}{0.93}
\begin{figure}[tb]
\centering
\psfrag{w}[cc][][0.6][90]{Waiting Time [s]}
\psfrag{i}[cc][][0.6][90]{Waiting Time Increase}
\psfrag{e}[cc][][0.9]{$\epsilon$}
\psfrag{1}[lc][][0.6]{$D=1$}
\psfrag{2}[lc][][0.6]{$D=2$}
\psfrag{3}[lc][][0.6]{$D=3$}
\psfrag{4}[lc][][0.6]{$D=4$}
\subfigure[250 vehicles]{\includegraphics[width=\fccc\columnwidth]{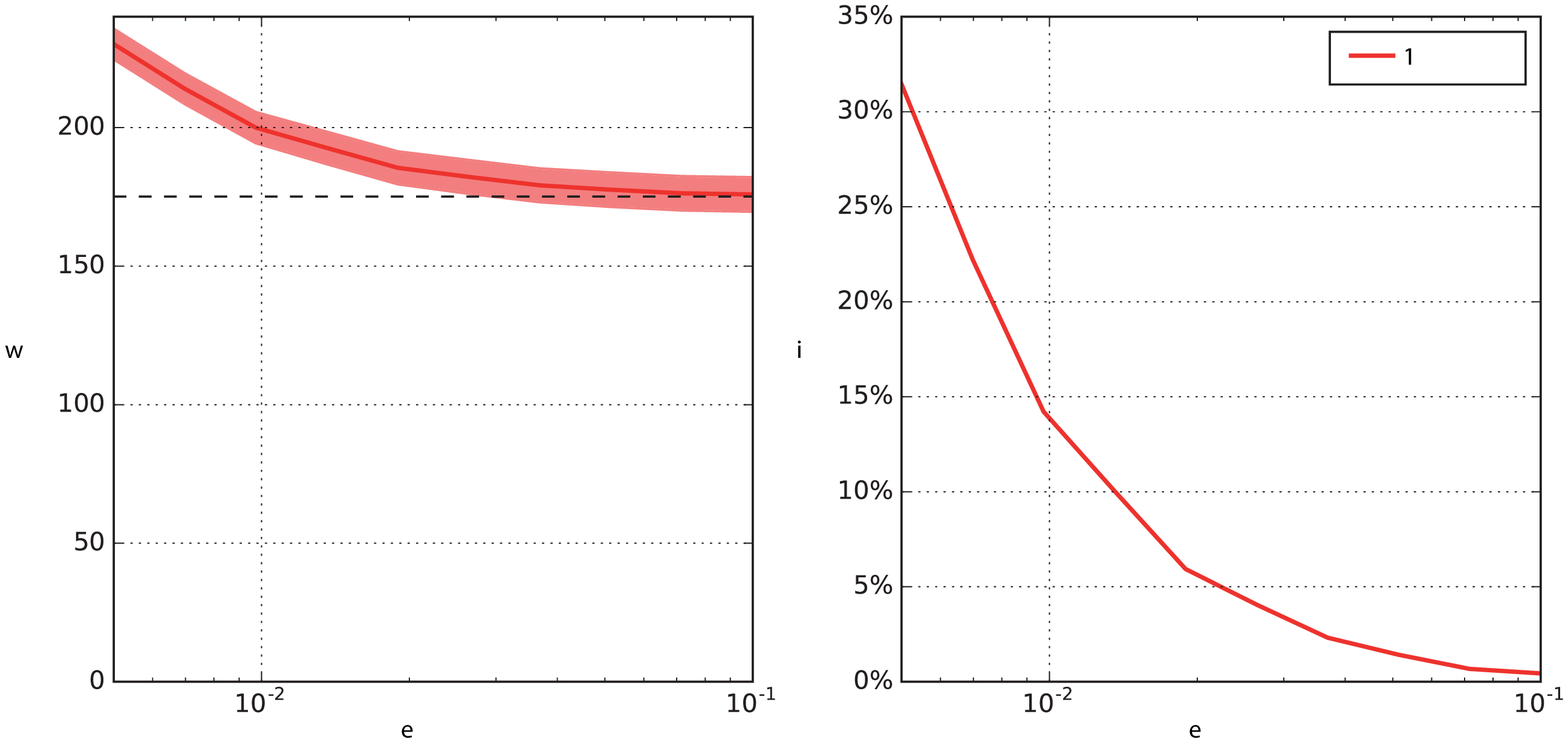}\label{fig:batch_waiting_time_a}}\\
\subfigure[1000 vehicles]{\includegraphics[width=\fccc\columnwidth]{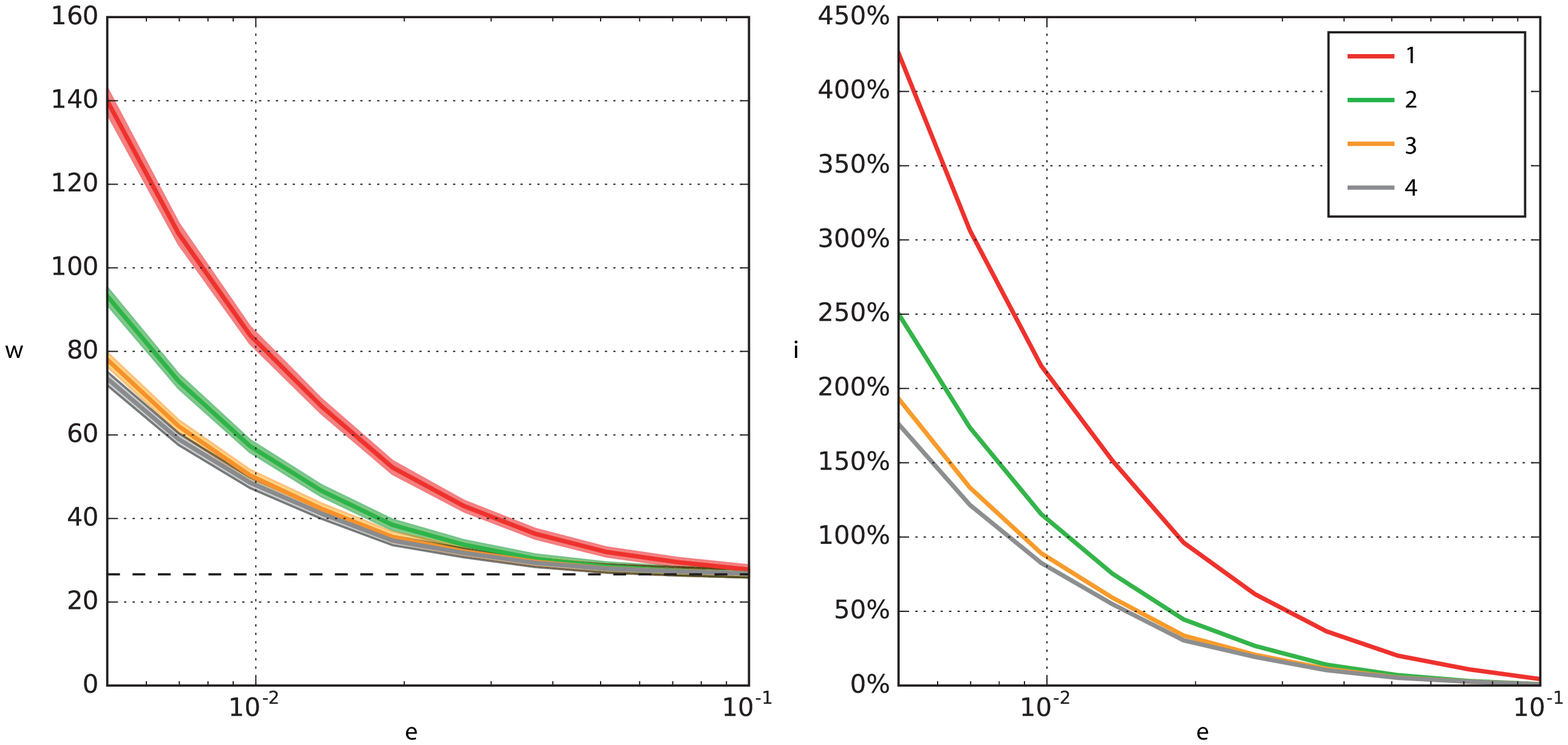}\label{fig:batch_waiting_time_c}}
\caption{Performance of the batch assignment algorithm for 250 passengers and a varying number of vehicles. Passenger and vehicle locations are sampled according to the actual distribution of pick-up and drop-off locations, respectively, over the month of June 2016. The left panels show the mean waiting time, as a function of the Laplacian inverse scale parameter $\epsilon$. The dashed line shows the optimal performance. The shaded areas represent a 95\% confidence interval. The right panels show the percentage of waiting time increase, with respect to the optimal assignment strategy. The value $D$ is the number of assigned vehicles per passenger, as elaborated in Sec.~\ref{sec:multi_vehicle_ass}
\label{fig:batch_waiting_time}}
\end{figure}

\begin{figure}[t]
\centering
\psfrag{i}[cc][][0.7]{$1$}
\psfrag{j}[cc][][0.7]{$2$}
\psfrag{k}[cc][][0.7]{$1,2$}
\psfrag{d}[cc][][0.8]{$\widetilde{t}_{1,\textrm{d}}$}
\psfrag{a}[cc][][0.8]{$t_{1,\textrm{p}}$}
\psfrag{c}[cc][][0.8]{$\widetilde{t}_{1,\textrm{p}}$}
\psfrag{b}[cc][][0.8]{$t_{1,\textrm{d}}$}
\psfrag{z}[cc][][0.8]{$\widetilde{t}_{1,\textrm{d}} = \widetilde{t}_{2,\textrm{d}}$}
\psfrag{e}[cc][][0.8]{$\widetilde{t}_{2,\textrm{p}}$}
\psfrag{f}[cc][][0.8]{$t_{2,\textrm{p}} + t_{2,\textrm{d}}$}
\subfigure[Operation diagram for continuous non-redundant assignment. Once assigned to a passenger, a vehicle will become available in $t=t_{\textrm{p}} + t_{\textrm{d}}$ seconds. However, to preserve the privacy mechanism, it will report its availability after $\tilde{t}=\widetilde{t}_{\textrm{p}} + \widetilde{t}_{\textrm{d}}$ seconds (irrespective of whether that time arrives earlier or later than its true availability).]{\includegraphics[width=.8\columnwidth]{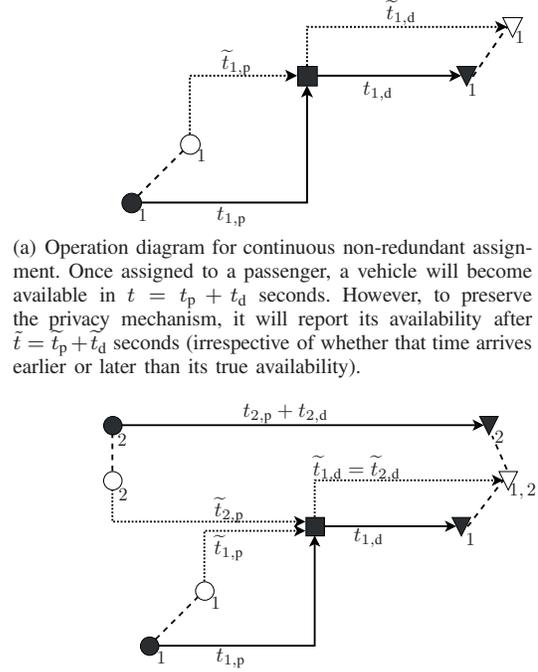}\label{fig:continuous_nonred}} \\
\subfigure[Operation diagram for continuous redundant assignment. Once assigned to a passenger, vehicles $1$ and $2$ mutually agree upon which vehicle will effectively pick up the passenger (in this case, vehicle $1$). 
After $\tilde{t}_1 = \widetilde{t}_{1,\textrm{p}} + \widetilde{t}_{1,\textrm{d}}$ seconds, the selected vehicle drops the passenger off at the request site $\btriangle_1$), but reports an obfuscated drop-off location $\wtriangle_1$. 
The other vehicle pretends to pick up the passenger by moving to $\btriangle_2$ and reporting the same drop-off as vehicle $1$ at $\wtriangle_1 = \wtriangle_2$ after $\tilde{t}_2 = \widetilde{t}_{2,\textrm{p}} + \widetilde{t}_{2,\textrm{d}}$ seconds.]{\includegraphics[width=.8\columnwidth]{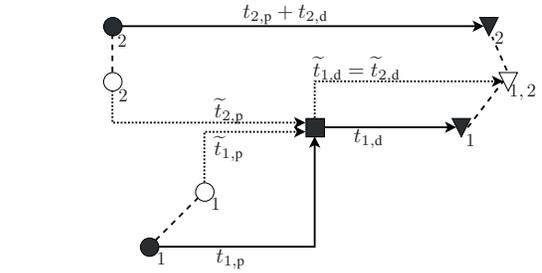}\label{fig:continuous_red}} 
\caption{Continuous assignment scheme. True positions are denoted by filled black symbols. Reported (obfuscated) positions are denoted by empty white symbols. The previous true and reported drop-off locations are shown with circles (i.e., $\bcircle$, $\wcircle$), the pick-up location is shown with a square (i.e., $\bsquare$) and the new true and reported drop-off locations are shown with triangles (i.e., $\btriangle$, $\wtriangle$, respectively). Note that all vehicles sample new offsets from the their true and reported positions according to a planar Laplace distribution with inverse-scale parameter $\epsilon$.
\label{fig:schema_continuous}}
\end{figure}

\section{Continuous Vehicle Routing under Privacy}
\label{sec:continuous_vrp}

In practice, after a vehicle has dropped off its passenger, it becomes available again for another assignment batch. We refer to consecutive assignments of the same vehicle to consecutive passengers as \emph{continuous} vehicle routing.
In contrast to batch vehicle routing, continuous vehicle routing poses the additional challenge of ensuring that the obfuscated drop-off locations are reported at \emph{times} that correspond to the travel distances between reported (obfuscated) locations. In other words, a vehicle effectively reports its availability at a moment in time that is either \emph{before} or \emph{after} it truly drops off its passenger, since reporting its availability at the true moment would compromise the privacy of the drop-off location.
In the following, we demonstrate that the continuous vehicle assignment strategy respects promised privacy guarantees. We elaborate the strategy for non-redundant as well as redundant assignments.

\subsection{Continuous Non-Redundant Vehicle Assignment}
The procedure according to which a vehicle is routed to a passenger in the private continuous assignment scheme is as follows. At the start, the vehicle communicates with the operator to report its obfuscated position, and to receive its next passenger assignment. Once this assignment is known, the vehicle directly communicates with the passenger (by-passing the central operator) to obtain the true passenger destination. Based on this information, the vehicle computes an obfuscated drop-off location (by adding planar Laplace noise) and the moment in time when this fictitious location will be reached (i.e., when the vehicle availability must be reported). Figure~\ref{fig:continuous_nonred} illustrates this procedure on a two-dimensional workspace, showing the offset produced by the privacy mechanism, and its effect on the travel time.
Since the vehicle availability, as reported to the operator, might happen before the actual vehicle availability, it is important that vehicles do not keep an ever increasing backlog of passenger requests. However, since obfuscated drop-off positions are sampled from unbiased probability distributions, there is no bias towards reporting availability sooner rather than later, and the backlog effect does not happen in practice.

For clarity, the following formulations use the symbols defined in the Fig.~\ref{fig:schema_continuous}, where $\bcircle$ and $\wcircle$ represent real and obfuscated origin locations, and where $\btriangle$ and $\wtriangle$ represent real and obfuscated drop-off locations, respectively.

\begin{proposition}[$\epsilon$-geo-indistinguishable continuous non-redundant assignment] \label{prop:2}
The continuous non-redundant assignment where each vehicle $i$ reports an obfuscated position $\wcircle_i$ drawn from a planar Laplace distribution is $\epsilon$-geo-indistinguishable with respect to the true positions $\bcircle_i$ iff, at each drop-off, each vehicle draws a new obfuscated position $\wtriangle_i$ from a planar Laplace distribution centered around the true drop-off $\btriangle_i$ (we assume that each passenger makes a single ride).
\end{proposition}
\begin{proof}
The proof follows a similar structure to the one of Proposition~\ref{prop:1}. Since the query only returns the last obfuscated position of all vehicles (and a passenger takes a single ride), we have:
\begin{equation}
\mathcal{L} = \!\mathop{\textrm{sup}} \left| \ln \frac{\mathbb{P}(\wcircle_1, \mdots, \wcircle_i, \wtriangle_i, \mdots, \wcircle_N, \widetilde{t} | \bcircle_1, \mdots, \bcircle_i, \btriangle_i, \mdots, \bcircle_N)}{\mathbb{P}(\wcircle_1, \mdots, \wcircle_i, \wtriangle_i, \mdots, \wcircle_N, \widetilde{t} | \bcircle_1, \mdots, \bcircle_i, \btriangle'_i, \mdots, \bcircle_N)} \right|
\end{equation}
where $\btriangle_i$ refers to the latest drop-off location of vehicle $i$, and $\btriangle'_i$ refers to an alternative drop-off location. The duration $\widetilde{t}$ (known to the operator) refers to the reported duration of the latest ride of vehicle $i$. As shown in Figure~\ref{fig:continuous_nonred}, this duration is fully determined by the previously and currently reported drop-off locations. Hence, if we assume independence of pick-up and drop-off locations, we obtain:
\begin{eqnarray}
\mathcal{L} &=& \mathop{\textrm{sup}} \left| \ln \frac{\mathbb{P}(\wcircle_i, \wtriangle_i, \widetilde{t} | \bcircle_i, \btriangle_i)}{\mathbb{P}(\wcircle_i, \wtriangle_i, \widetilde{t} | \bcircle_i, \btriangle'_i)} \right| \\
&=& \mathop{\textrm{sup}} \left| \ln \frac{\mathbb{P}(\widetilde{t} | \wcircle_i, \wtriangle_i) \mathbb{P}(\wcircle_i | \bcircle_i) \mathbb{P}(\wtriangle_i | \btriangle_i)}{\mathbb{P}(\widetilde{t} | \wcircle_i, \wtriangle_i) \mathbb{P}(\wcircle_i | \bcircle_i) \mathbb{P}(\wtriangle_i | \btriangle'_i)} \right| \nonumber \\
&=& \mathop{\textrm{sup}} \left| \ln \frac{\mathbb{P}(\wtriangle_i | \btriangle_i)}{\mathbb{P}(\wtriangle_i | \btriangle'_i)} \right| ~ \leq \epsilon \cdot \mathop{\textrm{sup}} \| \btriangle_i - \btriangle'_i \|_2 \nonumber
\end{eqnarray}
\end{proof}

\begin{remark}\label{rem:epsilon_tune}
This proof assumes the independence of pick-up and drop-off locations. In reality, it is often possible to correlate the pick-up and drop-off locations given the time of day. As a result, it may be necessary to vary the level of obfuscation throughout the day by tuning $\epsilon$ as a function of the pick-up location. 
\end{remark}
\begin{remark}\label{rem:epsilon_corr}
Passengers who take $n$ subsequent rides only benefit from an $n\epsilon$-geo-indistinguishable drop-off (since obfuscated positions are independent from each other). In practice, this leakage can be reduced by correlating subsequent obfuscated positions that relate to a given same passenger. 
\end{remark}

\subsection{Continuous Redundant Vehicle Assignment}
Much like the redundant batch assignment strategy, continuous assignment can also be implemented with a redundant number of vehicles per passenger. This procedure is schematized in Figure~\ref{fig:continuous_red}. In contrast to
Figure~\ref{fig:continuous_nonred}, Figure~\ref{fig:continuous_red} shows two vehicles that are assigned to pick up a passenger. The vehicles will mutually agree upon which one will truly pick up the passenger (i.e., the one that is truly closer). The selected vehicle computes an obfuscated drop-off location, and communicates this value to the redundant vehicle, which uses it to compute its itinerary (to a fake drop-off location). At the end of the respective travel times, both vehicles report their availability as well as the same obfuscated drop-off location (this operation is not synchronized).

\begin{proposition}[$\epsilon$-geo-indistinguishable continuous redundant assignment]
The continuous redundant assignment where each vehicle $i$ reports an obfuscated position $\wcircle_i$ drawn from a planar Laplace distribution is $\epsilon$-geo-indistinguishable with respect to the true positions $\bcircle_i$ iff vehicles that are assigned to the same passenger report the same drop-off location. At each drop-off, each vehicle draws a new obfuscated position (we assume that each passenger makes a single ride).
\end{proposition}
\begin{proof}
We illustrate the proof for a two-vehicle assignment where the vehicle $i$ picks up the passenger and drops her off at position $\btriangle_i$ (refer to Figure~\ref{fig:continuous_red}). Similarly to the previous proof, we obtain:

{\small
\begin{eqnarray}
&\mathcal{L}& = \mathop{\textrm{sup}} \left| \ln \frac{\mathbb{P}(\wcircle_i, \wcircle_j, \wtriangle_{i,j}, \widetilde{t}_i, \widetilde{t}_j | \bcircle_i, \bcircle_j, \btriangle_i)}{\mathbb{P}(\wcircle_i, \wcircle_j, \wtriangle_{i,j}, \widetilde{t}_i, \widetilde{t}_j | \bcircle_i, \bcircle_j, \btriangle'_i)} \right|   \\
&=& \!\!\!\!\! \mathop{\textrm{sup}} \left| \ln \frac{
    \mathbb{P}(\widetilde{t}_i | \wcircle_i, \wtriangle_{i,j})
    \mathbb{P}(\widetilde{t}_j | \wcircle_j, \wtriangle_{i,j})
    \mathbb{P}(\wcircle_i | \bcircle_i)
    \mathbb{P}(\wcircle_j | \bcircle_j)
    \mathbb{P}(\wtriangle_{i,j} | \btriangle_i)
  }{
    \mathbb{P}(\widetilde{t}_i | \wcircle_i, \wtriangle_{i,j})
    \mathbb{P}(\widetilde{t}_j | \wcircle_j, \wtriangle_{i,j})
    \mathbb{P}(\wcircle_i | \bcircle_i)
    \mathbb{P}(\wcircle_j | \bcircle_j)
    \mathbb{P}(\wtriangle_{i,j} | \btriangle'_i)
  } \right| \nonumber \\
&=& \!\!\!\!\! \mathop{\textrm{sup}} \left| \ln \frac{\mathbb{P}(\wtriangle_{i,j} | \btriangle_i)}{\mathbb{P}(\wtriangle_{i,j} | \btriangle'_i)} \right| \nonumber ~\leq \epsilon \cdot \mathop{\textrm{sup}} \| \btriangle_i - \btriangle'_i \|_j
\end{eqnarray}
}
The same holds if vehicle $j$ picks the passenger up.
\end{proof}

\subsection{Performance}
\label{sec:results_continuous}
The following results are based on the dataset and graph described in Section~\ref{sec:data}. We show the performance of the continuous assignment strategy (for a constant privacy level, given by $\epsilon=0.02$) applied to the data recorded during the 24 hours of Friday 1st, June 2016. 
The total vehicle fleet size is variable throughout time --- since the taxicab dataset contains records of occupied vehicles only, we make use of a heuristic to compute the total number of vehicles in the fleet (occupied plus available vehicles). This value is computed from the number of occupied rides obtained from the real taxi data, and is multiplied by 1.56 (which corresponds to a ratio of 64\% of occupied taxis~\footnote{\url{http://www.nyc.gov/html/tlc/downloads/pdf/2014_taxicab_fact_book.pdf}}). We cap the maximum fleet size at 6000 vehicles.
At each time step of our simulation, we ensure that the vehicle fleet size corresponds to the precomputed vehicle fleet size. When necessary, we add vehicles to the system --- these vehicles are placed at random locations that correspond to the distribution of drop-off locations derived from the real dataset. We use the real recorded passenger pick-up times to represent the times when vehicles are requested.
Requests are batched into 20\,s intervals. If requests are not serviced in the current batch, they roll over to the next one, and those that are not serviced within 20\,min are dropped (all schemes exhibit a drop rate below 0.01\%). Finally, we solve the assignment problem for each batch of requests. Vehicles are continuously routed according to our strategies, and they remain in the system unless removed (when the precomputed fleet size reduces).

Figure~\ref{fig:cont_waiting_times} shows the results in the form of violin plots that represent the distribution of the underlying data (i.e., one data point corresponds to the waiting time of a passenger until pick-up).
We show the performance of the continuous non-redundant and redundant assignment schemes, and, to benchmark our results, we also show the performance of the optimal (non-private) assignment algorithm. In the redundant scheme, for each batch, we choose the number of redundant assignments $D$ such there is at least 1 vehicle per passenger, and such that a sufficient number of vehicles is left unassigned (to account for the next batch of requests). The private non-redundant scheme is 30\% worse than the optimal assignment scheme with an average waiting time increase of 53\,s; the redundant scheme improves over the non-redundant scheme and is only 6\% worse that the optimal assignment scheme, with an average waiting time increase of 11\,s.

\begin{figure}[tb]
\centering
\psfrag{a}[cc][][0.7]{Non-Private}
\psfrag{e}[cc][][0.7]{Private Redundant}
\psfrag{c}[cc][][0.7]{Private Non-Red.}
\psfrag{t}[cc][][0.75][90]{Waiting Time [s]}
{\includegraphics[width=0.78\columnwidth]{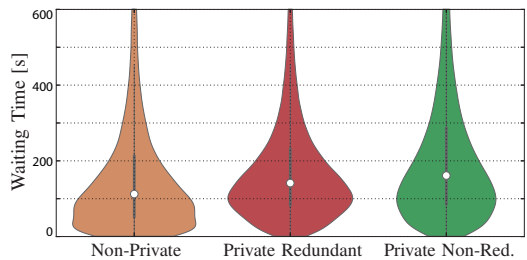}}
\caption{Performance of our assignment algorithms applied to the taxicab data recorded during the 24 hours of Friday 1st, June 2016, and processed as described in Sec.~\ref{sec:results_continuous}. The violin plots feature a kernel density estimation of the underlying data distribution, and are obtained using the Scott reference rule to compute the kernel bandwidth. The plots also show the median, the 25th and 75th percentiles. The mean and standard deviation for the non-private scheme are: $180.937 \pm 247.619$\,[s]; for the private redundant scheme: $191.381 \pm 194.309$\,[s]; and for the private non-redundant scheme: $233.929 \pm 263.267$\,[s]. 
\label{fig:cont_waiting_times}}
\end{figure}



\section{Conclusion} 
\label{sec:conclusion}
This work is situated in the context of centrally operated MoD systems, and considers the problem of assigning vehicles to passengers such that the mean waiting time is minimized. In specific, we provide a method that solves the assignment problem under obfuscated vehicle origin positions, such that the destinations of previously dropped-off passengers remains private.
Our main contributions are two-fold. First, we formalized the notion of privacy for continuous vehicle-to-passenger assignment by building on the concept of geo-indistinguishability. Second, to minimize performance loss, we presented an algorithm that takes advantage of superfluous vehicles in the system, combining multiple iterations of the Hungarian algorithm to allocate a redundant number of vehicles to a single passenger. We evaluated our method on a real, large-scale dataset consisting of over 11 million taxi rides (specifying vehicle availability and passenger requests).
Our results show that our privacy-preserving redundant assignment strategy successfully minimizes the loss of quality-of-service, as measured by average passenger waiting times.
Our work demonstrates that privacy can be integrated into MoD systems without incurring a significant loss of performance, and moreover, that this loss can be further minimized at the cost of deploying additional (redundant) vehicles into the fleet. 
Also, our results indicate that there is a trade-space between privacy and waiting time, and that this trade-off can be tuned as a function of system and/or user preferences. 

Future work will consider the integration obfuscated passenger pick-up locations (which can be readily obtained from the existing framework). We will also consider the tuning of individualized privacy levels as a function of user behavior (as addressed in Remarks~\ref{rem:epsilon_tune} and~\ref{rem:epsilon_corr}), or as a function of heterogeneous user preferences (e.g., reduction of waiting time at the cost of a loss of privacy).






\bibliographystyle{abbrvnat}
{\footnotesize
\bibliography{Bibliography}
}
\end{document}